\documentclass{CSML}
\pdfoutput=1

\usepackage{lastpage}
\newcommand{\dOi}{14(1:7)2018}
\lmcsheading{}{1--\pageref{LastPage}}{}{}%
{Jul.~13,~2016}{Jan.~17, 2018}{}

\usepackage{hyperref}
\hypersetup{hidelinks}
\usepackage[all]{xy}
\usepackage{graphicx,hyperref,amssymb,epsfig,epstopdf}

\numberwithin{equation}{section}

\newcommand{\Nat}{{\mathbb N}}

\newcommand{\lra}{\longrightarrow}

\newcommand{\da}{\mathord{\downarrow}}
\newcommand{\ua}{\mathord{\uparrow}}

\newcommand{\cl}{\operatorname{cl}}

\renewcommand{\S}{\mathcal{S}}
\newcommand{\T}{\mathcal{T}}
\newcommand{\dsup}{\mathop{\bigvee{}^{\,\makebox[0pt]{$\scriptstyle\uparrow$}}}}
\newcommand{\tda}{\mathord{\nabla}}
\renewcommand{\H}{{\mathcal H}}
\newcommand{\ninfty}{\overline{\Nat}}
\newcommand{\cn}[2]{{#1.#2}}

\newcommand{\tr}[3]{(#1,#2,#3)}
\newcommand{\ub}{\operatorname{ub}}
\newcommand{\sset}[1]{{\{#1\}}}
\newcommand{\set}[2]{\{{#1}\mid{#2}\}}
\newcommand{\sob}[1]{{#1}^s}
\newcommand{\osob}[1]{\widehat{#1}}
\newcommand{\cp}{\mathrm{irr}}
\newcommand{\el}{\varepsilon}
\newcommand{\sleq}{\sqsubseteq}
\newcommand{\sless}{\sqsubset}
\newcommand{\defeq}{\mathrel{\mathord{:}\mathord{=}}}
\begin{document}

\title[The Ho-Zhao problem]{The Ho-Zhao problem}

\author[W. K. Ho]{Weng Kin Ho}	
\address{National Institute of Education, Nanyang Technological University, Singapore}	
\email{wengkin.ho@nie.edu.sg}  
\thanks{The first author is partially supported by London Mathematical Society
  Computer Science Small Grants -- Scheme~7.}	

\author[J. Goubault-Larrecq]{Jean Goubault-Larrecq}	
\address{LSV, ENS Paris-Saclay, CNRS, Universit\'e Paris-Saclay, 94230 Cachan, France}	
\email{goubault@lsv.fr}  

\author[A. Jung]{Achim Jung}	
\address{School of Computer Science, The University of Birmingham, United Kingdom}	
\email{A.Jung@cs.bham.ac.uk}  

\author[X. Xi]{Xiaoyong Xi} 
\address{Department of Mathematics, Jiangsu Normal University, Xuzhou, China} 
\email{littlebrook@jsnu.edu.cn} 
\thanks{The fourth author is supported by the following research funds: (1) NSFC
  (National Science Foundation of China): 11671008; (2) NSFC: 1136028; and (3)
  NSF of project of Jiangsu Province of China (BK 20170483).} 





\begin{abstract}
\noindent Given a poset $P$, the set $\Gamma(P)$ of all Scott closed sets ordered by inclusion forms a complete lattice.  A subcategory $\mathbf{C}$ of $\mathbf{Pos}_d$ (the category of posets and Scott-continuous maps) is said to be $\Gamma$-faithful if for any posets $P$ and $Q$ in $\mathbf{C}$, $\Gamma(P) \cong \Gamma(Q)$ implies $P \cong Q$. It is known that the category of all continuous dcpos and the category of bounded complete dcpos are $\Gamma$-faithful, while $\mathbf{Pos}_d$ is not.  Ho \& Zhao (2009) asked whether the category $\mathbf{DCPO}$ of dcpos is $\Gamma$-faithful.  In this paper, we answer this question in the negative by exhibiting a counterexample.
To achieve this, we introduce a new subcategory of dcpos which is $\Gamma$-faithful.  This subcategory subsumes all currently known $\Gamma$-faithful subcategories.  With this new concept in mind, we construct the desired counterexample which relies heavily on Johnstone's famous dcpo which is not sober in its Scott topology.
\end{abstract}

\maketitle
\section{Introduction}
\label{sec: intro}
The collection, $\Gamma(X)$, of closed subsets of a topological space~$X$, ordered by inclusion, forms a distributive complete lattice often referred to as the \emph{closed set lattice} of~$X$.
If a lattice $L$ is isomorphic to the closed set lattice of some topological space~$X$, we say that $X$ is a \emph{topological representation} of~$L$.  It is natural to ask:
\begin{qu} \label{quest: first question}
Which lattices have topological representations?
\end{qu}
Seymour Papert was the first to characterize such lattices as those which are complete, distributive, and have a base consisting of irreducible elements,~\cite{papert59}.

We can also ask how much of the topological structure of a space is encoded in its closed set lattice. Following Wolfgang Thron, \cite{thron62}, we say that two topological spaces $X$ and~$Y$ are \emph{lattice-equivalent} if their closed-set lattices are order-isomorphic.  Clearly, homeomorphic spaces are lattice-equivalent, but the converse fails (even for $T_0$~spaces).
This then leads us to:
\begin{qu} \label{quest: second question}
Which classes $\mathbf{C}$ of topological spaces are such that any two lattice-equivalent spaces $X$ and~$Y \in \mathbf{C}$ are homeomorphic, i.e.:
\[
\forall X, Y \in \mathbf{C}.\;\Gamma(X) \cong \Gamma(Y) \implies X \cong Y\;?
\]
\end{qu}
Sober topological spaces are exactly those that can be fully reconstructed from their closed set lattices, \cite{drakethron65}; therefore the class $\mathbf{Sob}$ of sober spaces is a natural choice in answer to Question~\ref{quest: second question}. Furthermore, any topological space~$X$ is lattice-equivalent to its sobrification $\sob X$ and so it follows that $\mathbf{C}$ cannot contain a non-sober space~$X$ and its sobrification~$\sob X$ at the same time, in other words, $\mathbf{Sob}$ is a \emph{maximal} choice for~$\mathbf{C}$.

The two questions above can also be asked in the context of a particular class of topological spaces. The ones we have in mind were introduced by Dana Scott, \cite{scott72,gierzetal03}, and are known collectively as \emph{domains}. The characteristic feature of domains is that they carry a partial order and that their topology is completely determined by the order.
More precisely, let $P$ be a poset and $U$ a subset of~$P$. One says that $U$ is \emph{Scott open}, if (i)~$U$~is an upper set, and (ii)~$U$ is inaccessible by directed joins.  The set~$\sigma(P)$ of all Scott opens of $P$ forms the \emph{Scott topology} on~$P$, and $\Sigma P := (P,\sigma(P))$ is called the \emph{Scott space} of~$P$. In what follows, for a poset~$P$, we write $\Gamma(P)$ to always mean the lattice of Scott-closed subsets of~$P$.

We may now relativize our definitions to the context of Scott spaces. We say that a lattice~$L$ has a \emph{Scott-topological representation} if $L$ is isomorphic to $\Gamma(P)$ for some poset~$P$, and ask:
\begin{prob} \label{prob: first problem}
  Which lattices have Scott-topological representations?
\end{prob}
Although some work has been done on this problem, \cite{hozhao09}, as of now it remains open. In the special case of \emph{continuous domains} a very pleasing answer was given independently by Jimmie Lawson, \cite{lawson79}, and Rudolf-Eberhardt Hoffmann, \cite{hoffmann81}. They showed that a lattice~$L$ has a Scott-topological representation~$\Gamma(P)$ for some continuous domain~$P$ if and only if $L$ is completely distributive.

In the order-theoretic context the  second question reads as follows:
\begin{prob} \label{prob: second problem}
Which classes of posets $\mathbf{C}$ satisfy the condition
\[
\forall P,~Q \in \mathbf{C}.\;\Gamma(P) \cong \Gamma(Q) \implies P \cong Q\;?
\]
\end{prob}
A class $\mathbf{C}$ of posets is said to be \emph{$\Gamma$-faithful} if the above condition holds (\cite[p. 2170, Remark 2]{zhaofan10}).
The following classes of posets are known to be $\Gamma$-faithful. For the first, this follows directly from the definition of sobriety; the second statement was proven in \cite{hozhao09}.
\begin{enumerate}
\item $\mathbf{SOB}_\sigma$ of dcpos whose Scott topologies are sober, containing in particular $\mathbf{Cont}$, the class of all continuous dcpos;
\item $\mathbf{CSL}$ of complete semilattices (i.e., dcpos where all bounded subsets have a supremum),  containing in particular all complete lattices. \end{enumerate}

A \emph{dcpo-completion} of a poset $P$ is a dcpo $A$ together with a Scott-continuous mapping $\eta:P \lra A$ such that for any Scott-continuous mapping $f:P \lra B$ into a dcpo $B$ there exists a unique Scott-continuous mapping $\hat{f}:A \lra B$ satisfying $f = \hat{f} \circ \eta$.  It was shown in \cite{zhaofan10} that the dcpo-completion, $E(P)$, of a poset $P$ always exists; furthermore $\Gamma(P) \cong \Gamma(E(P))$.  Hence the class, $\mathbf{POS}$, of all posets is \emph{not} $\Gamma$-faithful.  Indeed any class of posets that is strictly larger than $\mathbf{DCPO}$, the class of dcpos, is not $\Gamma$-faithful.  This means that we can restrict attention entirely to dcpos, in particular, it leads one to ask:
\begin{prob}(\cite{hozhao09},~\cite[Remark 2]{zhaofan10})
\label{prob: Ho-Zhao problem}
Is $\mathbf{DCPO}$ $\Gamma$-faithful?
\end{prob}
This question was dubbed the \emph{Ho-Zhao problem} in \cite{alghousseinibrahim15}. The authors of this paper claimed that the two dcpos
\[
\Upsilon = ([0,1],\leq) \text{ and } \Psi = (\{[0,a] \mid 0 < a \leq 1\},\subseteq)
\]
show that $\mathbf{DCPO}$ is not $\Gamma$-faithful. However, it is easy to see that $\Psi \cong ((0,1],\leq)$ so that $\Gamma(\Psi) \cong ([0,1],\leq)$.  On the other hand, $\Gamma(\Upsilon) \cong ([0,1],\leq)_\bot$ (the unit interval plus an additional least element) which is not isomorphic to $\Gamma(\Psi)$.  This failure is only to be expected as both $\Upsilon$ and $\Psi$ are continuous dcpos and we already noted that $\mathbf{Cont}$ is $\Gamma$-faithful.


This paper comprises two parts. In Section~\ref{sec: positive} we present a positive result by introducing the class $\mathbf{domDCPO}$ of \emph{dominated dcpos} and showing it to be $\Gamma$-faithful.  Importantly, $\mathbf{domDCPO}$ subsumes all currently known $\Gamma$-faithful classes listed above.  In the second part (Sections~\ref{sec:johnstone's-countereg}--\ref{sec:hoxi-structure}) we show that the answer to the Ho-Zhao problem is negative.  We construct a dcpo~$\H$ which is not dominated, and derive from it a dominated dcpo~$\osob{\H}$ so that $\H \not \cong \osob{\H}$ but $\Gamma(\H) \cong \Gamma(\osob{\H})$.  The construction makes use of Johnstone's famous example of a dcpo~$\S$ whose Scott topology is not sober (\cite{johnstone81}).  To familiarize the reader with Johnstone's dcpo~$\S$, we revisit it in Section~\ref{sec:johnstone's-countereg}, highlighting its peculiarities.
This prepares us for the counterexample~$\H$, informally presented in Section~\ref{sec:hoxi-informal}.  Intrepid readers who are keen to pursue the detailed construction of $\H$ and how it works in answering the Ho-Zhao problem may then continue their odyssey into Section~\ref{sec:hoxi-structure}.

For notions from topology and domain theory we refer the reader to~\cite{abramskyjung94,gierzetal03,goubault13a}.

\section{A positive result}
\label{sec: positive}
\subsection{Irreducible sets}
A nonempty subset $A$ of a topological space $(X;\tau)$ is called \emph{irreducible} if whenever $A\subseteq B\cup C$ for closed sets $B$ and~$C$, $A\subseteq B$ or $A\subseteq C$ follows. We say that $A$ is \emph{closed irreducible} if it is closed and irreducible. The set of all closed irreducible subsets of~$X$ is denoted by~$\osob X$. The following two facts about irreducible sets follow immediately from the definition:
\begin{itemize}
\item If $A$ is irreducible then so is its topological closure.
\item The direct image of an irreducible set under a continuous function is again irreducible.
\end{itemize}
In this paper we are exclusively interested in dcpos and the Scott topology on them. When we say ``(closed) irreducible'' in this context we always mean (closed) irreducible with respect to the Scott topology. It is a fact that every directed set in a dcpo is irreducible in this sense, while a set that is directed and closed is of course a principal ideal, that is, of the form $\da x$. Peter Johnstone discovered in 1981, \cite{johnstone81}, that a dcpo may have irreducible sets which are not directed, and closed irreducible sets which are not principal ideals; we will discuss his famous example in Section~\ref{sec:johnstone's-countereg} and someone not familiar with it may want to have a look at it before reading on.

Given that we may view irreducible sets as a generalisation of directed sets, the following definition suggests itself:
\begin{defi}
  A dcpo $(D;\leq)$ is called \emph{strongly complete} if every irreducible subset of~$D$ has a supremum. In this case we also say that $D$ is an \emph{scpo}. A subset of an scpo is called \emph{strongly closed} if it closed under the formation of suprema of irreducible subsets.
\end{defi}

Note that despite this terminology, strongly complete partial orders are a long way from being complete in the sense of lattice theory.

\subsection{Categorical setting}
The basis of Questions \ref{quest: first question} and~\ref{quest: second question} in the Introduction is the well-known (dual) adjunction between topological spaces and frames, which for our purposes is expressed more appropriately as a (dual) adjunction between topological spaces and coframes:
\begin{displaymath}
  \xymatrix{ \makebox[2em][r]{$\Gamma:\mathbf{Top}$}\ar@/^/[r]_\bot & \makebox[2em][l]{$\mathbf{coFrm}^{\mathrm{op}}:\mathrm{spec}$}\ar@/^/[l]
    }
\end{displaymath}
Here $\mathrm{spec}$ is the functor that assigns to a coframe its set of \emph{irreducible} elements, that is, those elements~$a$ for which $a\leq b\vee c$ implies $a\leq b$ or $a\leq c$,\footnote{Strictly speaking, this is the definition for a \emph{coprime} element but in distributive lattices there is no difference between the two.} topologized by closed sets $B_b=\set{a\in\mathrm{spec}(L)}{a\leq b}$. Starting with a topological space $(X;\tau)$ we obtain the \emph{sobrification}~$\sob X$ of~$X$ by composing $\mathrm{spec}$ with~$\Gamma$. Concretely, the points of $\sob X$ are given by the closed irreducible subsets of~$X$ and the topology is given by closed sets $\sob B=\set{A\in\sob X}{A\subseteq B}$ where $B\in\Gamma(X)$.

In order to present Problems \ref{prob: first problem} and~\ref{prob: second problem} in a similar fashion, it seems natural to replace $\mathbf{Top}$ with the category of \emph{dcpo spaces}, i.e., directed-complete partially ordered sets equipped with the Scott topology. However, while we know that $\mathrm{spec}$ yields topological spaces which are dcpos in their specialisation order, it is \emph{not the case} that the topology on $\mathrm{spec}(L)$ equals the Scott topology with respect to that order; all we know is that every $\sob B$ is Scott-closed.

Rather than follow a topological route, therefore, we reduce the picture entirely to one concerning \emph{ordered sets}. To this end we restrict the adjunction above to the category $\mathbf{MCS}$ of \emph{monotone convergence spaces} (\cite[Definition~II-3.12]{gierzetal03}) on the topological side and compose it with the adjunction between monotone convergence spaces and dcpos
\begin{displaymath}
  \xymatrix{ \Sigma:\mathbf{DCPO}\ar@/^/[r]_\bot & \mathbf{MCS}:\mathrm{so}\ar@/^/[l]
    }
\end{displaymath}
which assigns to a dcpo its Scott space and to a space its set of points with the \emph{specialisation order}.
We obtain a functor from $\mathbf{DCPO}$ to~$\mathbf{coFrm}$ which assigns to a dcpo its coframe of Scott-closed subsets and we re-use the symbol~$\Gamma$ for it rather than writing $\Gamma\circ\Sigma$. In the other direction, we assign to a coframe the \emph{ordered set} of irreducible elements, where the order is inherited from the coframe. To emphasize the shift in perspective, we denote it with~$\cp$ rather than~$\mathrm{spec}$ or $\mathrm{so}\circ\mathrm{spec}$. It will also prove worthwhile to recall the action of $\cp$ on morphisms: If $h\colon L\to M$ is a coframe homomorphism, then $\cp(h)$ maps an irreducible element~$a$ of~$M$ to $\bigwedge\set{x\in L}{h(x)\geq a}$.
Altogether we obtain the following (dual) adjunction:
\begin{displaymath}
  \xymatrix{ \makebox[2em][r]{$\Gamma:\mathbf{DCPO}$}\ar@/^/[r]_\bot & \makebox[2em][l]{$\mathbf{coFrm}^{\mathrm{op}}:\cp$}\ar@/^/[l]
    }
\end{displaymath}
Its unit~$\eta_D$ maps an element $x$ of a dcpo~$D$ to the principal downset $\da x$, which is always closed and irreducible and hence an element of $\cp(\Gamma(D))$. The counit $\varepsilon_L$ (as a concrete map between coframes) sends an element~$b$ of a coframe~$L$ to the set $\set{a\in\cp(L)}{a\leq b}$ which is clearly closed under directed suprema, hence an element of $\Gamma(\cp(L))$.

Combining $\cp$ with~$\Gamma$ yields a monad on $\mathbf{DCPO}$ which we denote with $\osob{(\;)}$. Concretely, it assigns to a dcpo~$D$ the set~$\osob D$ of closed irreducible subsets ordered by inclusion. We call this structure the \emph{order sobrification of~$D$}. If $f\colon D\to E$ is a Scott-continuous function between dcpos, and $A\subseteq D$ is a closed irreducible set, then $\osob f(A)$ is the Scott closure of the direct image~$f(A)$ (which is again irreducible as we noted at the beginning of this section). The monad unit is given by the unit of the adjunction mentioned before. Following through the categorical translations one sees that multiplication $\mu_D=\cp(\epsilon_{\Gamma D})\colon\osob{\osob D}\to\osob D$ maps a closed irreducible collection\footnote{To help the reader we will usually use the words ``collection'' or ``family'' rather than ``set'' when referring to subsets of~$\osob D$.}  of closed irreducible subsets to the closure of their union. However, one can show that the closure operation is not needed and we recall the proof in a moment. For now let us stress that this order-theoretic monad is \emph{not} idempotent which is an important difference to its topological counterpart, the sobrification monad. We will illustrate this in Section~\ref{sec:hoxi-informal} with the help of Johnstone's non-sober dcpo.

In preparation for the calculations that follow, let us explore concretely some of the ingredients of the order sobrification monad.
\begin{prop}\label{prop:union}
  Let $D$ be a dcpo.
  \begin{enumerate}
  \item If $B$ is a closed set of $D$ then $\varepsilon_{\Gamma(D)}(B)=\set{A\in\osob D}{A\subseteq B}\in\Gamma(\osob D)$.
  \item If $\mathcal{A}$ is irreducible as a subset of~$\osob D$ then $\bigcup\mathcal{A}$ is irreducible as a subset of~$D$.
  \item If $\mathcal{A}$ is Scott-closed as a subset of~$\osob D$ then $\bigcup\mathcal{A}$ is Scott-closed as a subset of~$D$.
  \item $\osob D$ is strongly complete.
  \item If $B$ is a closed set of $D$ then $\varepsilon_{\Gamma(D)}(B)$ is strongly closed.
  \end{enumerate}
\end{prop}
\begin{proof}
  (1) This follows from the general description of the counit.

  (2) Assume we have $\bigcup\mathcal{A}\subseteq B\cup C$ with $B,C$ Scott-closed subsets of~$D$. Then every member of $\mathcal{A}$ must belong to either $B$ or~$C$ because of their irreducibility and so it follows from (1) that $\mathcal{A}\subseteq\varepsilon_{\Gamma(D)}(B)\cup\varepsilon_{\Gamma(D)}(C)$. Furthermore, $\varepsilon_{\Gamma(D)}(B)$ and $\varepsilon_{\Gamma(D)}(C)$ are Scott-closed subsets of~$\osob D$, so by assumption, one of them must cover the irreducible collection~$\mathcal{A}$, say $\varepsilon_{\Gamma(D)}(B)$. It then follows that $\bigcup\mathcal{A}$ is contained in~$B$.

  (3) If $S$ is a directed subset of $\bigcup\mathcal{A}$ then each $x\in S$ must belong to some $A_x\in\mathcal{A}$. Because each $A_x$ is a lower set of~$D$ we have $\da x\subseteq A_x$ and because $\mathcal{A}$ is a lower set of $\osob D$ it follows that $\da x\in\mathcal{A}$. The collection $(\da x)_{x\in S}$ is directed and its supremum~$\da\dsup S$ belongs to~$\mathcal{A}$ as $\mathcal{A}$ is assumed to be Scott-closed. If follows that $\dsup S$ belongs to $\bigcup\mathcal{A}$.

  (4) Let $\mathcal{A}$ be an irreducible subset of $\osob D$. Then by (2) we have that $\bigcup\mathcal{A}$ is an irreducible subset of~$D$, whence its closure~$A$ is an element of~$\osob D$. Clearly $A$ is the supremum of~$\mathcal{A}$.

  (5) follows from (4) because $B$ is a dcpo in its own right and $\varepsilon_{\Gamma(D)}(B)$ is the same poset as~$\osob B$.
\end{proof}

\begin{lem}\label{lem:strong}
  Let $D$ be a dcpo and $\mathcal{B}\in\Gamma(\osob D)$.
  \begin{enumerate}
  \item $\mathcal{B}\subseteq\varepsilon_{\Gamma(D)}(\bigcup\mathcal{B})$.
  \item If $\mathcal{B}$ is strongly closed then equality holds.
  \end{enumerate}
\end{lem}
\begin{proof}
  The first statement is trivial by the definition of the counit~$\varepsilon$. For the second, let $A$ be a closed irreducible subset of $\bigcup\mathcal{B}$. We need to show that $A$ is an element of~$\mathcal{B}$. Every element~$x$ of~$A$ belongs to some $A_x\in\mathcal{B}$. As we argued in part~(3) of the preceding proof, it follows that for every $x\in A$, $\da x\in\mathcal{B}$. We claim that the collection $\mathcal{A}=\set{\da x}{x\in A}$ is irreducible as a subset of~$\osob D$. This will finish our proof as we clearly have that $A$ is the supremum of~$\mathcal{A}$ and by assumption, $\mathcal{B}$ is closed under forming suprema of irreducible subsets.

  So let $\mathcal{A}$ be covered by two closed collections $\mathcal{M},\mathcal{N}\in\Gamma(\osob D)$, in other words, every $\da x$, $x\in A$, belongs to either $\mathcal{M}$ or~$\mathcal{N}$. It follows that each $x\in A$ belongs to either $\bigcup\mathcal{M}$ or~$\bigcup\mathcal{N}$ and these two sets are Scott-closed by part~(3) of the preceding proposition. Because $A$ is irreducible, it is already covered by one of the two, and this implies that $\mathcal{A}$ is covered by either $\mathcal{M}$ or~$\mathcal{N}$.
\end{proof}

\subsection{Question~\ref{prob: second problem} revisited}
We approach Question~\ref{prob: second problem} via the monad $\osob{(\;)}$. Starting from the assumption $\Gamma(D)\cong\Gamma(E)$ we immediately infer $\osob D=\cp(\Gamma(D))\cong\cp(\Gamma(E))=\osob E$ and the question then becomes whether this isomorphism implies $D\cong E$. Our counterexample will demonstrate that in general the answer is ``no'' but in this section we will exhibit a new class~$\mathbf{domDCPO}$ of \emph{dominated dcpos} for which the answer is positive, that is, we will show:
\[
\forall D,~E \in \mathbf{domDCPO}.\;\osob D \cong \osob E \implies D \cong E\;.
\]
Before we do so, let us check that invoking the monad does not change the original question, in other words, the assumption $\osob D\cong\osob E$ is neither stronger nor weaker than $\Gamma(D)\cong\Gamma(E)$:
\begin{prop}
  For arbitrary dcpos $D$ and~$E$, $\;\;\osob D\cong\osob E\iff\Gamma(D)\cong\Gamma(E)$.
\end{prop}
\begin{proof}
  The implication from right to left is trivial, so assume we are given an order isomorphism $i\colon\osob D\to\osob E$. The idea for an isomorphism~$\phi$ from $\Gamma(D)$ to $\Gamma(E)$ is very simple: Given $B\in\Gamma(D)$ we compute $\varepsilon_{\Gamma(D)}(B)$, the collection of all closed irreducible sets contained in~$B$ (which belongs to $\Gamma(\osob D)$ by \ref{prop:union}(1)). Each of these can be replaced with its counterpart in~$E$ via the given isomorphism~$i$. In~$E$, then, we simply take the union of the collection $i(\varepsilon_{\Gamma(D)}(B))$. Using the maps that are provided to us by the adjunction, we can express $\phi$ as follows:
  \begin{displaymath}
    \xymatrix{ \phi\colon\Gamma(D)\ar[r]^{\varepsilon_{\Gamma(D)}} & \Gamma(\osob D)\ar[r]^{\Gamma(i^{-1})} & \Gamma(\osob E)\ar[r]^{\bigcup} & \Gamma(E) }
  \end{displaymath}
  For an inverse, we follow the same steps, starting at $\Gamma(E)$:
  \begin{displaymath}
    \xymatrix{ \psi\colon\Gamma(E)\ar[r]^{\varepsilon_{\Gamma(E)}} & \Gamma(\osob E)\ar[r]^{\Gamma(i)} & \Gamma(\osob D)\ar[r]^{\bigcup} & \Gamma(D) }
  \end{displaymath}
  In order to show that these are inverses of each other we use the fact that $\varepsilon_{\Gamma(D)}(B)$ is \emph{strongly} closed which we established in Proposition~\ref{prop:union}(5). Since the concept of strong closure is purely order-theoretic we get that the direct image under~$i$ is again strongly closed. This is crucial as it allows us to invoke Lemma~\ref{lem:strong}(2). The computation thus reads:
  \begin{displaymath}
    \begin{array}{rcll}
      \psi\circ\phi &=& \bigcup\circ \Gamma(i)\circ \varepsilon_{\Gamma(E)}\circ  \bigcup\circ \Gamma(i^{-1})\circ \varepsilon_{\Gamma(D)} & \text{by definition}\\
      &=& \bigcup\circ \Gamma(i)\circ \Gamma(i^{-1})\circ \varepsilon_{\Gamma(D)} & \text{Lemma~\ref{lem:strong}(2)}\\
      &=& \bigcup\circ \varepsilon_{\Gamma(D)} & \text{$\Gamma$ is a functor}\\
      &=& \mathrm{id}_{\Gamma(D)}
    \end{array}
  \end{displaymath}
  The last equality follows from the fact that $\varepsilon_{\Gamma(D)}(B)$ contains all sets of the form $\da x$, $x\in B$.

  The other composition, $\phi\circ\psi$, simplifies in exactly the same way to the identity on~$\Gamma(E)$, and since all maps involved are order-preserving, we have shown that the pair~$\phi, \psi$ constitutes an order isomorphism between $\Gamma(D)$ and~$\Gamma(E)$.
\end{proof}

\subsection{Dominated dcpos}
Our new version of Question~\ref{prob: second problem} requires us to infer $D\cong E$ from $\osob D\cong\osob E$, and the most direct approach is to find a way to recognize \emph{purely order-theoretically} inside~$\osob D$ those elements which correspond to closed irreducible subsets of the form~$\da x$ with $x\in D$. As our counterexample~$\H$ to the Ho-Zhao problem will show, this is not possible for general dcpos. The purpose of the present section is to exhibit a class of dcpos for which the direct approach works.

\begin{defi}
  Given $A',A\in\osob D$ we write $A'\triangleleft A$ if there is $x\in A$ such that $A'\subseteq\da x$. We write $\tda A$ for the set $\set{A'\in\osob D}{A'\triangleleft A}$.
\end{defi}

Clearly, an element $A$ of~$\osob D$ is of the form $\da x$ if and only if $A\triangleleft A$ holds, but this is not yet useful since the definition of~$\triangleleft$ makes explicit reference to the underlying dcpo~$D$. We can, however, record the following useful facts:

\begin{prop}\label{prop:nabla-irr}
  Let $D$ be a dcpo and $A\subseteq D$ be closed and irreducible.
  \begin{enumerate}
  \item $A=\bigvee\nabla A$
  \item $\nabla A$ is irreducible as a subset of~$\osob D$.
  \end{enumerate}
\end{prop}
\begin{proof}
  The first statement is trivial because $\nabla A$ contains all principal ideals $\da x$, $x\in A$. The proof of the second statement is essentially the same as that of Lemma~\ref{lem:strong}(2).
\end{proof}

\begin{defi}
  Let $D$ be a strongly complete partial order and $x',x\in D$. We write $x'\prec x$ if for all closed irreducible subsets~$A$ of~$D$, $x\leq\bigvee A$ implies $x'\in A$. We say that $x\in D$ is \emph{$\prec$-compact} if $x\prec x$, and denote the set of $\prec$-compact elements by $K(L)$.
\end{defi}

Note that this is an \emph{intrinsic} definition of a relation on~$D$ without reference to any other structure. It is reminiscent of the way-below relation of domain theory but note that it is defined via \emph{closed} irreducible sets. This choice has the following consequence, which is definitely not true for way-below:

\begin{prop}\label{prop:prec-downset}
  Let $D$ be an scpo and $x\in D$. The set $\set{a\in D}{a\prec x}$ is Scott-closed.
\end{prop}
\begin{proof}
  Let $(a_i)_{i\in I}$ be a directed set of elements, each of which is $\prec$-below~$x$. We need to show that $\dsup_{i\in I}a_i$ is also $\prec$-below~$x$. To this end let $A$ be a closed irreducible subset of~$D$ with $x\leq \bigvee A$. Since $a_i\prec x$ for all $i\in I$ we have that every $a_i$ belongs to~$A$ and because $A$ is closed, $\dsup_{i\in I}a_i\in A$ follows.
\end{proof}

Recall from Proposition~\ref{prop:union}(4) that $\osob D$ is strongly complete for any dcpo~$D$, so on $\osob D$ we can consider both $\triangleleft$ and~$\prec$. We observe:

\begin{prop}
  $A'\triangleleft A$ implies $A'\prec A$ for all $A,A'\in\osob D$.
\end{prop}
\begin{proof}
  This holds because the supremum of a closed irreducible collection of closed irreducible subsets is given by union as we saw in Proposition~\ref{prop:union}(2) and~(3).
\end{proof}

Our aim now is to give a condition for dcpos~$D$ which guarantees the reverse implication.

\begin{defi}
  A dcpo~$D$ is called \emph{dominated} if for every closed irreducible subset~$A$ of~$D$, the collection $\nabla A$ is Scott-closed in~$\osob D$.
\end{defi}

We are ready for the final technical step in our argument:

\begin{lem}
  A dcpo~$D$ is dominated, if and only if $A'\prec A$ implies $A'\triangleleft A$ for all $A,A'\in\osob D$.
\end{lem}
\begin{proof}
  If: Let $A\in\osob D$. By assumption we have $\nabla A=\set{A'\in\osob D}{A'\triangleleft A} = \set{A'\in\osob D}{A'\prec A}$ and in Proposition~\ref{prop:prec-downset} we showed that the latter is always Scott-closed.

  Only if: Let $A'\prec A$. We know by Proposition~\ref{prop:nabla-irr}(2) that $\nabla A$ is irreducible; it is closed by assumption. We also know that $A=\bigvee\nabla A$ so it must be the case that $A'\in\nabla A$.
\end{proof}

We are ready to reap the benefits of our hard work. All of the following are now easy corollaries:

\begin{prop}
  For a dominated dcpo $D$, the only $\prec$-compact elements of~$\osob D$ are the principal ideals~$\da x$, $x\in D$. Also, the unit~$\eta_D$ is an order isomorphism from $D$ to~$K(\osob D)$.
\end{prop}

\begin{thm}\label{dominated->faithful}
  Let $D$ and $E$ be dominated dcpos. The following are equivalent:
  \begin{enumerate}
  \item $D \cong E$.
  \item $\Gamma(D) \cong \Gamma(E)$.
  \item $\osob{D} \cong \osob{E}$.
  \end{enumerate}
\end{thm}

\begin{thm}
  The class $\mathbf{domDCPO}$ of dominated dcpos is $\Gamma$-faithful.
\end{thm}

\subsection{Examples}
Let us now explore the reach of our result and exhibit some better known classes of dcpos which are subsumed by $\mathbf{domDCPO}$.

\begin{thm}
  \label{thm:dominated}
  The following are all dominated:
  \begin{enumerate}
  \item strongly complete partial orders;
  \item complete semilattices;
  \item complete lattices;
  \item dcpos which are sober in their Scott topologies;
  \item $\osob D$ for any dcpo~$D$.
  \end{enumerate}
\end{thm}
\begin{proof}
  (1) If $A'\triangleleft A$ for closed irreducible subsets of an scpo, then by definition $A'\subseteq\da x$ for some $x\in A$ and hence $\bigvee A'\in A$. If $(A_i)_{i\in I}$ is a directed family in $\nabla A$, then $\dsup_{i\in I}A_i\subseteq\da\dsup_{i\in I}(\bigvee A_i)$, and the element $\dsup_{i\in I}(\bigvee A_i)$ belongs to~$A$ because $A$ is a Scott-closed set.

  (2) Complete semilattices are dcpos in which every bounded subset has a supremum. Now note that in the proof of (1) we only needed suprema of irreducible sets which are contained in a principal ideal~$\da x$.

  (3) This is a special case of~(2).

  (4) In a sober dcpo the only closed irreducible subsets are principal ideals, so every sober dcpo is strongly complete and the claim follows from~(1).

  (5) We showed in Proposition~\ref{prop:union}(4) that $\osob D$ is always strongly complete.
\end{proof}

Recall that a topological space is called \emph{coherent} if the intersection of any two compact saturated sets is again compact. It is called \emph{well-filtered} if whenever $(K_i)_{i\in I}$ is a filtered collection of compact saturated sets contained in an open~$O$, then some~$K_i$ is contained in~$O$ already.

\begin{lem}\label{lem: ub(A) is compact saturated}
  Let $D$ be a dcpo which is well-filtered and coherent in its Scott topology. Then for any nonempty $A \subseteq D$, the set~$\ub(A)$ of upper bounds of~$A$ is compact saturated.
\end{lem}
\begin{proof}
  Since $\ub(A) = \bigcap_{a \in A} \ua a$, it is saturated. For any finite nonempty $F \subseteq A$, $\bigcap_{a \in F} \ua a$ is compact since $D$ is coherent. Thus $\{\bigcap_{a \in F} \ua a \mid F \subseteq_{\rm finite} A\}$ forms a filtered family of compact saturated sets whose intersection equals~$\ub(A)$. If it is covered by a collection of open sets then by well-filteredness of~$D$ some $\bigcap_{a \in F} \ua a$ is covered already, but the latter is compact by coherence so a finite subcollection suffices to cover it.
\end{proof}

\begin{prop}\label{prop: filtered and coherent implies dominated}
  Every dcpo~$D$ which is well-filtered and coherent in its Scott topology is dominated.
\end{prop}
\begin{proof}
  Let $(A_i)_{i\in I}$ be a directed family in $\nabla A$. By the preceding proposition, the sets $\ub{A_i}$ are compact saturated and form a filtered collection. Suppose for the sake of contradiction that $\bigcap_{i \in I} \ub(A_i) \subseteq D \backslash A$.  Since $D \backslash A$ is Scott open, by the well-filteredness of $D$ it follows that there exists $i_0 \in I$ such that $\ub(A_{i_0}) \subseteq D \backslash A$. This contradicts the fact that every $A_i$ is bounded by an element of~$A$.
\end{proof}

At the juncture, the curious reader must be wondering why we have not yet given an example of a dcpo which is not dominated. Now, as we showed, any dcpo whose Scott topology is sober is dominated and so a non-dominated dcpo, if it exists, must be non-sober. Peter Johnstone, \cite{johnstone81} was the first to give an example of such a dcpo and it is thus natural to wonder whether it is dominated or not. In order to answer this question we explore his example~$\S$ in some detail in the next section. This will also help us to construct and study our counterexample to the Ho-Zhao conjecture.

\section{Johnstone's counterexample revisited}
\label{sec:johnstone's-countereg}
We denote with $\ninfty$ the \emph{ordered set} of natural numbers
augmented with a largest element~$\infty$. When we write~$\Nat$ we
mean the natural numbers as a set, that is, \emph{discretely ordered}.
Johnstone's counterexample~$\S$, depicted in
Figure~\ref{fig:Johnstone's-non-sober-dcpo}, is based on the ordered
set $\Nat \times \ninfty$, that is, a countable collection of infinite
chains. We call the chain $C_m := \sset m\times\ninfty$ the
\emph{$m$-th component} of~$\S$. Similarly, for a fixed
element~$n\in\ninfty$ we call the set $L_n := \Nat\times\sset n$ the
\emph{$n$-th level}. So $(m,n)$ is the unique element in the $m$-th
component on level~$n$. We call the elements on the $\infty$-th level
\emph{limit points}, and refer to all the others as \emph{finite
  elements}.

The order on $\S$ is given by the product order plus the
stipulation that the limit point~$(m,\infty)$ of the $m$-th component
be above every element of levels $1$ to~$m$. We express this formally as follows (where $m,m',n,n'$ are natural numbers):
\begin{itemize}
\item $(m,n)<_1(m,n')$ if $n<n'$
\item $(m,n)<_2(m,\infty)$
\item $(m,n)<_3(m',\infty)$ if $n\leq m'$.
\end{itemize}
A simple case distinction shows that ${<}\defeq{<_1}\cup{<_2}\cup{<_3}$ is transitive and irreflexive, hence ${\leq}\defeq({<}\cup{=})$ is an order relation on~$\S$.

\begin{figure}[t]
\includegraphics[scale=0.8]{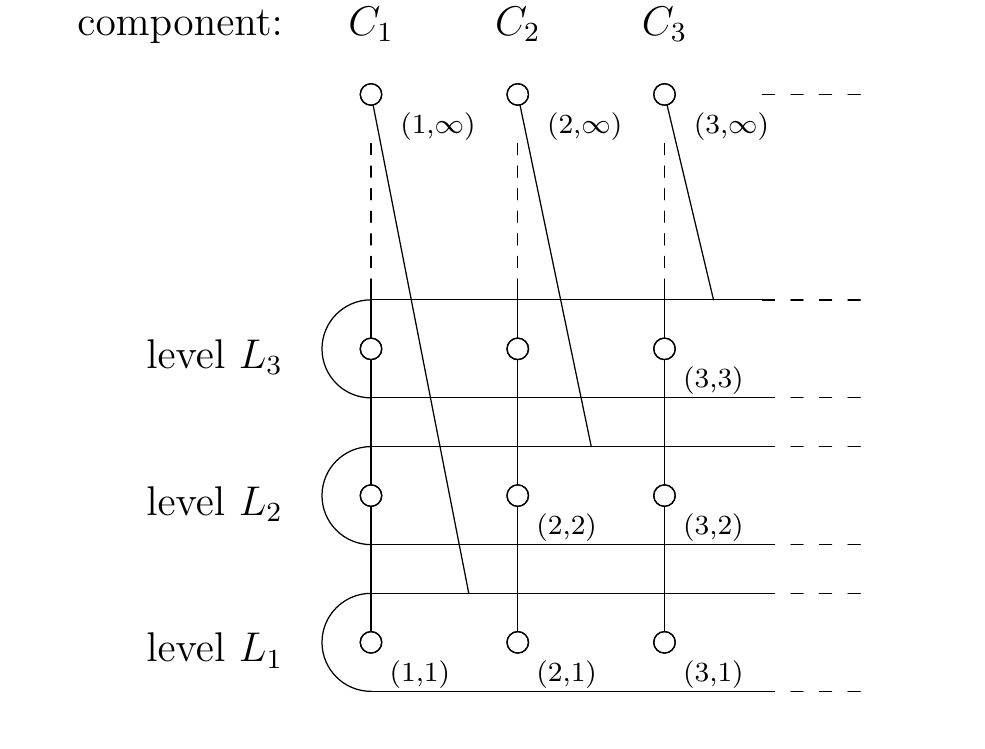}
\caption{\label{fig:Johnstone's-non-sober-dcpo}Johnstone's non-sober
  dcpo $\S$.}
\end{figure}

As the Hasse diagram makes clear, the only non-trivial (i.e., not containing their supremum) directed sets
of~$\S$ are the chains contained in a component, with supremum the
component's limit point. Since every finite element~$(m,n)$ is also
below the limit points~$(m',\infty)$, where $n\leq m'$, but not below
any other finite element outside its own component, we see that none
of them is \emph{compact} in the sense of domain theory. It follows
that $\S$ is highly non-algebraic, and indeed it couldn't be algebraic
as algebraic dcpos are always sober spaces in their Scott topology.

Let us take a closer look at the Scott topology on~$\S$. The two
defining conditions of Scott-closed sets manifest themselves in the
following properties:
\begin{itemize}
\item Closed sets are lower sets: if the set contains a limit
  point~$(m,\infty)$, then it must contain the component~$C_m$ and all
  levels $L_1,\ldots,L_m$.
\item Closed sets are closed under the formation of limits: if the set
  contains infinitely many elements of any one component then it must
  contain the limit point of that component.
\end{itemize}
Taken together, we obtain the following principle (which plays a
crucial role in our construction as well):
\begin{itemize}
\item[$(\dagger)$] If a Scott-closed subset of~$\S$ contains
  infinitely many limit points, then it equals~$\S$.
\end{itemize}
This is because such a set contains infinitely many levels by the
first property of Scott-closed sets, which means that it contains
infinitely many elements of every component, and therefore contains
the limit points of all components by the second property. Applying
the first property again we see that it contains everything.

It is now easy to see that $\S$ itself is an irreducible closed set:
If we cover $\S$ with two closed subsets then at least one of them
must contain infinitely many limit points. By $(\dagger)$ that set
then is already all of~$\S$. Of course, $\S$ is not the closure of a
singleton since there is no largest element, so we may conclude, as
Peter Johnstone did in \cite{johnstone81}, that $\S$ is not sober in
its Scott topology. In the terminology of the last section, it also
follows that $\S$ is not strongly complete. In preparation for our own counterexample, let us prove that $\S$ is the \emph{only} closed subset which is not a point closure.

\begin{prop}\label{p:irrS}
  The closed irreducible subsets of~$\S$ are $\S$ itself and the
  closures of singleton sets.
\end{prop}
\begin{proof}
  Assume that $A$ is a closed irreducible proper subset of $\S$. We distinguish
  two cases:

  Case 1: $A$ contains no limit points. Let
  $(m,n)\in\Nat\times\Nat$ be a maximal element of~$A$. We claim
  that $B=A\setminus\da(m,n)$ is Scott-closed. Indeed, since $A$
  contains no infinite chains at all we don't need to worry about
  closure under limits. Downward closure holds because there is
  no order relationship between finite elements from different
  components of~$\S$. We therefore have the decomposition
  $A=B\cup\da(m,n)$ which shows that $A$ can only be irreducible if
  $B\subseteq\da(m,n)$ and $A=\da(m,n)=\cl(m,n)$.

  Case 2: $A$ contains limit points. Let $L$ be the non-empty set of those. By the principle~$(\dagger)$, $L$ is a finite set. Furthermore, let $B$ be the set of maximal elements of~$A$ which are not limit points. As in the previous paragraph, one sees that $\da B$ is a closed set. Hence we have the decomposition $A=\da B\cup\bigcup_{x\in L}\da x$ into finitely many closed sets. Irreducibility implies that $A$ is equal to one of the $\da x$ with $x\in L$.
\end{proof}

\begin{prop}\label{prop: Johnstone is dominated}
  Johnstone's dcpo $\S$ is dominated.
\end{prop}
\begin{proof}
  The only non-trivial directed sets in $\osob\S$ are the chains $(\da(m,n))_{n\in\Nat}$ and their supremum is $\da(m,\infty)$. If a closed irreducible subset~$A$ of~$\osob\S$ contains such a chain it must also contain its supremum, and because the supremum is a principal downset, it belongs to $\nabla A$.
\end{proof}

\begin{rem}
By Theorem~\ref{thm:dominated}(5), $\osob{\S}$ is dominated.  From the preceding result, $\S$ is dominated.  Hence we can invoke Theorem~\ref{dominated->faithful} and conclude that $\Gamma(\S) \not \cong \Gamma(\osob{\S})$ since $\S \not\cong \osob{\S}$.  Thus the pair of dcpos $\S$ and $\osob{\S}$ is not a counterexample to the Ho-Zhao conjecture.
\end{rem}

\begin{figure}[t]
\includegraphics[scale=0.8]{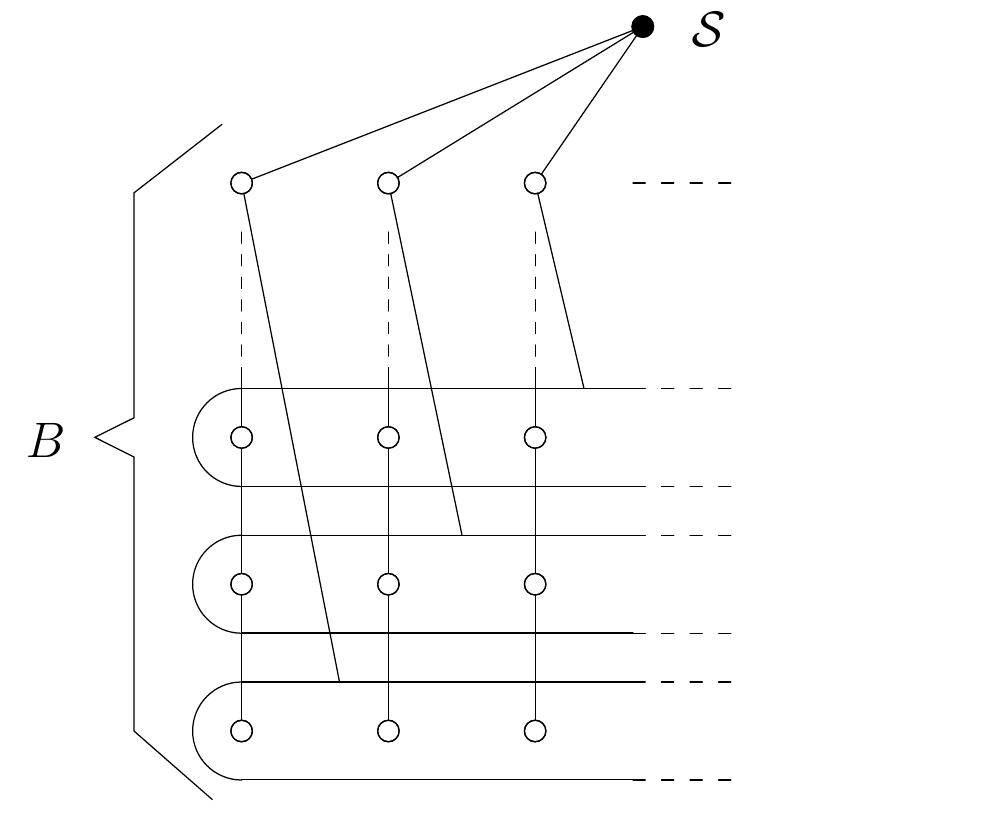}
\caption{\label{fig:sobrification-of-johnstone}The order sobrification
  $\osob\S$ of Johnstone's non-sober dcpo.}
\end{figure}

Proposition~\ref{p:irrS} allows us to draw the Hasse diagram of $\osob\S$, as we have done in Figure~\ref{fig:sobrification-of-johnstone}. We see that the extra point at the top of
$\osob S$ cannot be reached by a directed set; it is \emph{compact} in
the sense of domain theory. Also, the image of~$\S$ under~$\eta_\S$
forms a Scott-closed subset~$B$ of~$\osob\S$. Together this means that
the largest element of $\Gamma(\osob\S)$ is compact, while the largest
element of $\Gamma(\S)$ is not: it is the directed limit of the closed
sets $\da L_n$, $n\in\Nat$.

As an aside, we may also observe that the order sobrification of~$\osob\S$ would add yet another point, placed between the elements of~$B$ and the top element~$\osob \S$ which shows that order sobrification is not an idempotent process. In fact, our construction of~$\H$ addresses exactly this point,
by making sure that the new elements that appear in the order sobrification
are not compact and do not lead to new Scott-closed subsets.

\section{An informal description of the counterexample~$\H$}
\label{sec:hoxi-informal}
The construction of $\H$ may be viewed as an infinite process. We
begin with $\S$ and for every finite level~$L_n$
of~$\S$, we add another copy~$\S_n$, whose infinite
elements are identified with the elements of~$L_n$. No order relation
between the \emph{finite} elements of two different
$\S_n$,~$\S_{n'}$ is introduced. Now the process is
repeated with \emph{each} finite level of \emph{each}~$\S_n$,
adding a further $\Nat\times\Nat$ many copies of~$\S$. We keep
going like this \emph{ad infinitum} and in this way ensure all
elements are limit elements.

To make this a bit more precise, let $\Nat^*$ be the set of strings of
natural numbers. We write $\varepsilon$ for the empty string, $\cn ns$ for the result of adding the element~$n$ to the front of
string~$s$, and $ts$ for the concatenation of $t$ and~$s$.

We use $\Nat^*$ to index the many copies of Johnstone's example that
make up our dcpo~$\H$. In a first step, we let $\H'$ be the disjoint
union of all $\S_s$, $s\in\Nat^*$. We label individual elements
of~$\H'$ with triples $\tr mns\in\Nat\times\ninfty\times\Nat^*$ in the
obvious way. On $\H'$ we consider the equivalence relation~$\sim$
which identifies the finite element $\tr mns$ of~$\S_s$ with the
infinite element $\tr m\infty{\cn ns}$ of~$\S_{\cn ns}$, for all
$m,n\in\Nat$ and $s\in\Nat^*$. We may now define $\H$ as the quotient
of~$\H'$ by~$\sim$. The order on $\H$, which we will define formally in the next session, can be thought of as the quotient order,
that is, the smallest preorder such that the quotient map from $\H'$
to~$\H$ is monotone.

Furthermore, it can be shown that the Scott topology on $\H$ is the quotient topology of the Scott topology on~$\H'$. From this it follows that the characterisation of
Scott-closed subsets given in Section~\ref{sec:johnstone's-countereg}
above is still valid. The principle~$(\dagger)$ that we derived from
this, however, now has greater reach: Since every finite level in any
copy of~$\S$ is simultaneously the set of limit points of a subsequent copy,
it holds that a Scott-closed subset of~$\H$ cannot contain infinitely
many elements of \emph{any} level without containing all of them. A
close analysis of the situation (which we carry out in the next
section) then shows that $\H$ has one closed irreducible subset,
which is not a point closure, for \emph{every} $s\in\Nat^*$, to wit,
the set $\da L_s$.

\begin{figure}[t]
\includegraphics[scale=0.8]{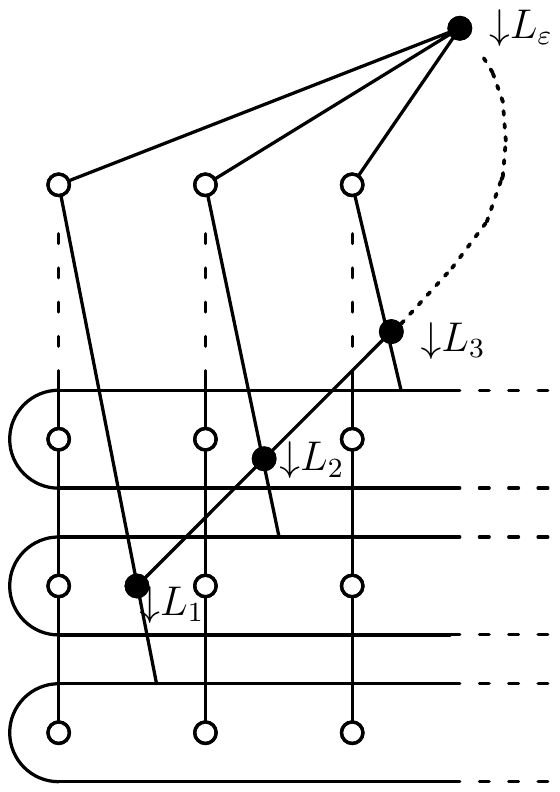}
\caption{\label{fig:sober-hoxi}The top part of the order sobrification
  of~$\H$.}
\end{figure}

The (inclusion) order among the $\da L_s$ is exactly as indicated in
Figure~\ref{fig:tree} which implies in particular that none of them
is compact in the sense of domain theory when viewed as an element of
the order sobrification~$\osob\H$. More precisely, we have that $\da L_s$ is
the limit (i.e., the closure of the union) of the chain~$\da L_{\cn
  1s}\subseteq \da L_{\cn 2s}\subseteq \da L_{\cn 3s}\subseteq\cdots$.

If we look just at the top part in the
order sobrification~$\osob\H$ then we obtain the structure displayed in
Figure~\ref{fig:sober-hoxi}. The set~$B$ of point closures (indicated
as open circles in Figure~\ref{fig:sober-hoxi}) no longer forms a
Scott-closed subset of~$\osob\H$: Because a closed set must be
downward closed, $\overline B$ contains all $\da L_n$, $n\in\Nat$, and
because it must be closed under taking limits, it contains
$\da L_{\el }$ as well, which means it equals all of~$\osob\H$.

The same considerations hold throughout $\osob\H$, and from this we will be able to conclude that the Scott topology of $\osob\H$ is the
same as the sobrification topology, which in turn is always isomorphic
to the topology of the original space, in our case the Scott topology
of~$\H$. Thus our analysis will show:

\begin{thm}\label{thm:H=sobH}
  The Scott topologies of $\H$ and~$\osob\H$ are isomorphic.
\end{thm}

On the other hand, $\H$ and~$\osob\H$ are clearly not isomorphic
as ordered sets: the latter has a largest element whereas the former
does not. Thus we have:

\begin{cor}\label{cor:DCPO-not-faithful}
  The category $\mathbf{DCPO}$ is not $\Gamma$-faithful.
\end{cor}

\section{Formal arguments regarding $\H$}
\label{sec:hoxi-structure}
\subsection{The structure of $\H$ as an ordered set}
\label{subsec: hoxi as a poset}
Above, we introduced $\H$ as a quotient structure of $\Nat^*$-many copies of~$\S$. However, since each equivalence class contains exactly one infinite element $(m,\infty,s)$ it is more straightforward to work with these representatives rather than classes. In view of this we abbreviate $(m,\infty,s)$ to $(m,s)$, that is, we work with $\Nat\times\Nat^*$ as the underlying set of~$\H$.

In analogy to the definition of the order on~$\S$, Section~\ref{sec:johnstone's-countereg}, we define the following relations on $\Nat\times\Nat^*$ (where $m,m',n,n'\in\Nat$, $t,s\in\Nat^*$):\footnote{Recall our notational conventions regarding strings from Section~\ref{sec:hoxi-informal}.}
\begin{itemize}
\item $(m,\cn ns)<_1(m,\cn {n'}s)$ if $n<n'$
\item $(m,ts)<_2(m,s)$ if $t\ne\el$
\item $(m,ts)<_3(m',s)$ if $t\ne\el$ and $\min(t)\leq m'$
\end{itemize}

The following properties are straightforward to check (we use ``;'' for relation composition):
\begin{prop}\label{prop:char<}\leavevmode
  \begin{enumerate}
  \item $<_1$, $<_2$, and $<_3$ are transitive and irreflexive.
  \item ${<_1};{<_2}={<_2}$
  \item ${<_1};{<_3}\subseteq{<_3}$
  \item ${<_2};{<_3}\subseteq{<_3}$
  \item ${<_3};{<_2}\subseteq{<_3}$
  \item ${<}\defeq{<_1}\cup{<_2}\cup{<_3}\cup({<_2};{<_1})\cup({<_3};{<_1})$ is transitive and irreflexive.
  \item ${\leq}\defeq({<}\cup{=})$ is an order relation.
  \end{enumerate}
\end{prop}

The set $\Nat\times\Nat^*$ together with the order relation~$\leq$ is $\H$, the poset which we claim to be a counterexample to the Ho-Zhao conjecture.

\begin{rem}
  Given a fixed string~$s$, we can define the subset $S_s\defeq\set{(m,s)}{m\in\Nat}\cup\set{(m,\cn ns)}{m,n\in\Nat}$ of~$\Nat\times\Nat^*$. Under the assignment $(m,s)\mapsto(m,\infty)$, $(m,\cn ns)\mapsto(m,n)$ we obtain that the restriction of $<_1$, $<_2$, and $<_3$ to~$S_s$ corresponds precisely to the equally named relations that we used to define the order on~$\S$ at the beginning of Section~\ref{sec:johnstone's-countereg}. Thus we obtain order-isomorphic copies~$\S_s$ of~$\S$ inside~$\H$ for every $s\in\Nat^*$, which links our definition of~$\H$ to the discussion in the previous section. No further use will be made of this fact, however.
\end{rem}

\subsection{The structure of $\H$ as a dcpo}
\label{subsec: hoxi as a dcpo}
The set $\Nat\times\Nat^*$ is countable; therefore every directed set of~$\H$ contains a cofinal $\Nat$-indexed chain and we can restrict attention to the latter when discussing suprema of directed sets, the Scott topology, and Scott-continuous functions. We call such a chain \emph{non-trivial} if it is strictly increasing and hence does not contain a largest element. In fact, we can restrict further:

\begin{prop}\label{prop:cofinal chain}
  Every non-trivial chain in~$\H$ contains a cofinal chain of the form $(m,\cn ns)_{n\in N}$ where $m\in\Nat$ and $s\in\Nat^*$ are fixed and $N$ is an infinite subset of~$\Nat$.
\end{prop}
\begin{proof}
  Let $(m_1,s_1)<(m_2,s_2)<\cdots$ be a non-trivial chain. Each relationship $(m_i,s_i)<(m_{i+1},s_{i+1})$ has to be one of the five types listed in item~(6) of Proposition~\ref{prop:char<}. All of these, except $<_1$, strictly reduce the length of the string~$s_i$, so they can occur only finitely often along the chain. Therefore, from some index~$i_0$ onward, the connecting relationship must always be~$<_1$ which implies that the shape of the entries $(m_i,s_i)$, $i\geq i_0$, is as stated.
\end{proof}

\begin{prop}\label{p:H-is-a-dcpo}
  $\H$ is a dcpo. More precisely, the supremum of a non-trivial chain $(m,\cn ns)_{n\in N}$, $N\subseteq\Nat$, is $(m,s)$.
\end{prop}
\begin{proof}
  It is clear that all elements in the chain are related to $(m,s)$ by $<_2$, so $(m,s)$ is an upper bound. If $(m',s')$ is another upper bound then the chain elements must be related to it by one of the five types listed in item~(6) of Proposition~\ref{prop:char<}. At least one of the five types must be used infinitely often and we can proceed by case distinction: Consider $<_1$; this implies that $(m',s')$ has shape $(m,\cn ks)$ with $k$ larger than all $n\in N$, which is impossible since $N$ is an infinite subset of~$\Nat$. If $(m,\cn ns)<_2(m',s')$ infinitely often (or even just once) then $m=m'$ and $s'$ is a suffix of $s$; hence either $(m,s)$ is equal to $(m',s')$, or $(m,s)<_2(m',s')$.

Next consider the case where $(m,\cn ns)<_3(m',s')$ holds infinitely often. Again, $s'$ must be a suffix of $s$, say $s=ts'$, and $m'\geq\min(\cn nt)$ for all $n\in N$. Because $N$ is infinite, we must have $m'\geq\min(t)$ and in particular, $t\ne\el$. It follows that $(m,s)<_3(m',s')$.

The final two cases, ${<_2};{<_1}$ and ${<_3};{<_1}$, are similar to $<_2$ and $<_3$, respectively.
\end{proof}

\subsection{The Scott topology on~$\H$}
\label{subsec: Scott topology on hoxi}
Propositions \ref{prop:cofinal chain} and~\ref{p:H-is-a-dcpo} allow us to characterize the Scott-closed subsets of~$\H$ as follows:

\begin{prop}
  A subset~$A$ of ~$\H$ is Scott-closed if it is downward closed and if it contains $(m,s)$ for every non-trivial chain $(m,\cn ns)_{n\in N}$ contained in~$A$.
\end{prop}

Given a string $s\in\Nat^*$, we call the set $L_s\defeq\set{(m,s)}{m\in\Nat}$ a \emph{level} of~$\H$. We have the following analogue to the principle~$(\dagger)$:
\begin{itemize}
\item[$(\ddagger)$] If a Scott-closed subset of~$\H$ contains
  infinitely many elements of a level~$L_s$, then
  it contains all of~$L_s$.
\end{itemize}
The proof is as before: Since $A$ contains infinitely many elements $(n,s)$, $n\in N\subseteq\Nat$, it contains infinitely many levels $L_{\cn ns}$, $n\in N$, because the elements of $L_{\cn ns}$ are below $(n,s)$ by relation~$<_3$. This implies that for any $m\in\Nat$ we have the non-trivial chain $(m,\cn ns)_{n\in N}$ in~$A$. Because $A$ is Scott-closed, the supremum $(m,s)$ must also also belong to~$A$. We get that $L_s\subseteq A$.

Alternately, we can express $(\ddagger)$ by saying that $L_s$ together with the restriction of the Scott topology is homeomorphic to $\Nat$ with the cofinite topology.

\subsection{The irreducible subsets of~$\H$, part~1}
\label{subsec: irreducible set in hoxi}
We are ready to exhibit certain Scott-closed subsets of $\H$ as irreducible:

\begin{prop}\label{p:characterisation}
  The following are closed irreducible subsets of $\H$:
  \begin{enumerate}
  \item downsets of individual elements of~$\H$;
  \item downsets of levels~$L_s$.
  \end{enumerate}
\end{prop}
\begin{proof}
  Sets of the first kind can also be seen as closures of singleton
  subsets and such sets are always irreducible. Sets of the second
  kind are clearly irreducible by principle~$(\ddagger)$. The only
  concern is whether they are closed. For this we argue in a fashion similar to the proof of Proposition~\ref{p:H-is-a-dcpo}: Let $(m,\cn nu)_{n\in N}$ be a non-trivial chain contained in $\da L_s$. Each element $(m,\cn nu)$ of the chain is strictly below some element $(m',s)\in L_s$. We once again invoke item~(6) of Proposition~\ref{prop:char<} to conclude that one of the five types listed there must occur infinitely often. It can't be~$<_1$ because $N$ is unbounded. If it is $<_2$, then either $u=s$ and the supremum $(m,u)$ belongs to the level~$L_s$ itself, or $s=tu$ in which case $(m,u)<_2(m,s)\in L_s$. The case $<_3$ is a bit more interesting: if we have $(m,\cn nu)<_3(m',s)$ infinitely often (where $m'$ is allowed to depend on~$n$), then $s$ must be equal to~$u$ or a proper suffix of it. In the first case, the supremum $(m,u)$ belongs to $L_s$ and in the second case, $(m,u)<_2(m,s)\in L_s$.

The arguments for the final two cases, ${<_2};{<_1}$ and ${<_3};{<_1}$, are similar to those for $<_2$ and $<_3$, respectively.
\end{proof}

We will soon show that the sets listed above are the \emph{only} closed irreducible subsets of~$\H$. At this point, however, we want to deliver on our promise, made at the end of Section~\ref{sec: positive}, of giving an example of a non-dominated dcpo.

\begin{prop}
\label{prop: hoxi is not dominated}
The dcpo $\H$ is not dominated.
\end{prop}
\begin{proof}
By the preceding proposition, $\H=\da L_\el\in \osob{\H}$ and for every $n \in \Nat$, $\da L_{(n)} \in \osob{\H}$ (where we use ``$(n)$'' for the singleton list containing~$n$).  Furthermore, for each $n \in \Nat$, $\da L_{(n)} \subseteq \da (n,\el)$ and hence $\mathcal{F} := \{\da L_{(n)} \mid n \in \Nat\}$ is a directed family in $\nabla\H$.  But $\bigvee \mathcal{F}$ equals $\H$ and it is not the case that $\H\triangleleft\H$.  Thus, $\H$ fails to be dominated.
\end{proof}

\subsection{From the set $\Nat^*$ to the tree~$\T$}
\label{subsec:tree}
\begin{figure}[t]
\includegraphics[scale=0.8]{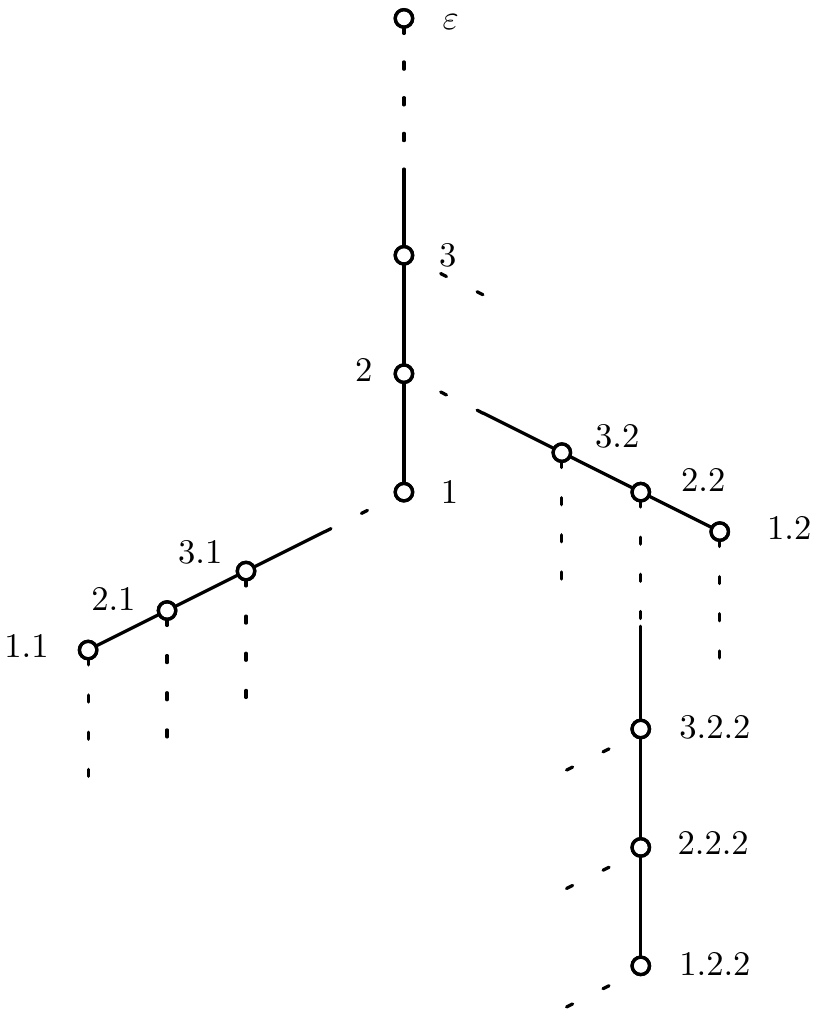}
\caption{\label{fig:tree}The order~$\sleq$ on $\Nat^*$.}
\end{figure}

For the definition of the order structure on~$\H$ in Section~\ref{subsec: hoxi as a poset} we introduced three relations, $<_1$, $<_2$, and~$<_3$. The first two of these concern only the string component of points $(m,s)$. We use them to define an order~$\sleq$ on~$\Nat^*$:
\begin{itemize}
\item $\cn ns\sless_1\cn{n'}s$ if $n<n'$
\item $ts\sless_2 s$ if $t\ne\el$
\end{itemize}
In analogy to Proposition~\ref{prop:char<}, one may now show that ${\sless}\defeq{\sless_1}\cup{\sless_2}\cup({\sless_2};{\sless_1})$ is transitive and irreflexive, and hence that ${\sleq}\defeq({\sless}\cup{=})$ is an order relation on~$\Nat^*$. Figure~\ref{fig:tree} attempts to give an impression of the resulting ordered set~$\T$.

Clearly, if $s'\sleq s$ in $\T$ then $(m,s')\leq(m,s)$ in~$\H$, for any $m\in\Nat$. Moreover, we have:

\begin{prop}\label{p:tree}
  $\T$ is a tree with root~$\el$, specifically, $\ua s$ is linearly ordered for all $s\in\Nat^*$.\footnote{This is a slight abuse of terminology: For a poset to qualify as a ``tree'' it is common to also require that $\sqsupseteq$ is well-founded. This is not the case here.}
\end{prop}
\begin{proof}
  Let $s'$ and~$s''$ be two different elements strictly above~$s$. We use the definition of~$\sless$ and consider the possible combinations: If $s\sless_1s'$ and $s\sless_1s''$ then either $s'$ or $s''$ has increased the first element of~$s$ by more than the other one and then the two are related themselves by~$\sless_1$. Similarly, if $s\sless_2s'$ and $s\sless_2s''$ then either $s'$ or $s''$ has dropped a longer initial segment of~$s$ and the two are related themselves by~$\sless_2$. If $s\sless_1s'$ and $s\sless_2s''$ then $s'\sless_2s''$. The cases involving ${\sless_2};{\sless_1}$ are similar.
\end{proof}

In analogy to Propositions \ref{prop:cofinal chain} and~\ref{p:H-is-a-dcpo} one shows that one can restrict attention to those chains which are of the form $(\cn ns)_{n\in N}$ where $N$ is an infinite subset of~$\Nat$; for these the supremum is the element~$s$. All in all, we have:

\begin{prop}\label{p:level function l}
  The poset $\T$ is a dcpo and the map $l\colon\H\to\T$ defined by $l(m,s)=s$ is Scott-continuous. Moreover, $\T$ is sober in the Scott topology.
\end{prop}
\begin{proof}
  Only the last statement is still in need of justification. Let $A$ be a closed irreducible subset of~$\T$ and $s$ a maximal element of~$A$. Then the set $A\setminus\da s$ is Scott-closed: let $t\in A\setminus\da s$ and $t'\sleq t$. If $t'\sleq s$ then either $s\sless t$ or $t\sless s$ by Proposition~\ref{p:tree}. The former is not possible since $s$ is maximal in~$A$ and the latter violates the assumption. Hence $t'\not\sleq s$ and therefore $t'\in A\setminus\da s$. The supremum of a non-trivial chain contained in $A\setminus\da s$ is itself contained in $A$ since $A$ is closed, and it can't be below~$s$ or otherwise all chain elements would be below~$s$. So indeed, $A\setminus\da s$ is Scott-closed. Since we have $A=(A\setminus\da s)\cup\da s$, irreducibility implies that $A=\da s$.
\end{proof}

\subsection{The irreducible subsets of~$\H$, part~2}
\label{subsec:and only if}
We now want to show that besides the closed irreducible sets listed in Proposition~\ref{p:characterisation} there are no others. We begin by examining how a closed subset~$A$ of~$\H$ intersects with the levels~$L_s$. For this we identify $L_s$ with $\Nat$ and set $A_s\defeq\set{m\in\Nat}{(m,s)\in A}$. We already know from principle~$(\ddagger)$ that each $A_s$ is either a finite set or all of~$\Nat$. Because of the close connection between the orders on $\H$ and~$\T$ we can furthermore state that $s'\sleq s$ implies $A_s\subseteq A_{s'}$, or in other words, the assignment $s\mapsto A_s$ is an antitone function from $\T$ to the powerset of~$\Nat$. Writing $\#(N)$ for the cardinality of a set, we have the following crucial fact:

\begin{lem}\label{l:key}
  The set $T_{\geq k}^A\defeq\set{s\in\Nat^*}{\#(A_s)\geq k}$ is Scott-closed in~$\T$.
\end{lem}
\begin{proof}
  We already know that $T_{\geq k}^A$ is a lower set since $s\mapsto A_s$ is antitone. For closure under suprema, let $(\cn ns)_{n\in N}$ be a non-trivial chain contained in $T_{\geq k}^A$. If all the $A_{\cn ns}$ are equal to~$\Nat$, then for any $m\in\Nat$ we get the chain $(m,\cn ns)_{n\in N}$ in~$A$ whose supremum $(m,s)$ also must also belong to~$A$. In other words, in this case $L_s\subseteq A$ and hence $s\in T_{\geq k}^A$. If one of the $A_{\cn ns}$ is a finite set, then there is $n_0\in N$ such that $A_{\cn ns}$ is the \emph{same} finite set~$M$ for all $n\geq n_0$. For each $m\in M$ we have the chain $(m,\cn ns)_{n\in N, n\geq n_0}$ contained in~$A$, whence the supremum $(m,s)$ also belongs to~$A$. We obtain that $A_s$ is also equal to~$M$, and since $\#(M)\geq k$ we have $s\in T_{\geq k}^A$.
\end{proof}

We can use this lemma to learn more about the function $l\colon(m,s)\mapsto s$ which we introduced in Proposition~\ref{p:level function l}:

\begin{cor}\label{c:closed}
  The function~$l\colon\H\to\T$ is closed.
\end{cor}
\begin{proof}
  Let $A$ be a Scott-closed subset of~$\H$. Since $l(A)=T^A_{\geq1}$ the lemma implies that $l(A)$ is closed in~$\T$.
\end{proof}

\begin{thm}\label{thm:characterisation}
  The only irreducible closed subsets of $\H$ are those listed in Proposition~\ref{p:characterisation}, that is, closures of single elements and closures of levels.
\end{thm}
\begin{proof}
  Let $A$ be an irreducible Scott-closed subset of~$\H$. Then $l(A)$ is irreducible because $l$ is continuous (Proposition~\ref{p:level function l}) and Scott-closed because $l$ is closed (Proposition~\ref{c:closed}). Since $\T$ is sober, $l(A)$ is equal to the Scott-closure of an element~$s\in\T$. Now we distinguish two cases: If $A_s=\Nat$ then $L_s\subseteq A$ and hence $A=\da L_s$. If, on the other hand, $\#(A_s)=k\in\Nat$, then we may consider the subset $T_{\geq k+1}^A$ of~$\T$ which is closed by Lemma~\ref{l:key}. By construction, $s$ does not belong to $T_{\geq k+1}^A$, so $B\defeq l^{-1}(T_{\geq k+1}^A)$ (which is closed by Proposition~\ref{p:level function l}) is definitely not all of~$A$. Furthermore, any element~$(m',s')$ of~$A$ which does not belong to $\da(A\cap L_s)$ lies in~$B$; this is because it is in addition to the $k$-many elements $\set{(m,s')}{m\in A_s}$ which do belong to $\da(A\cap L_s)$. This means that we can write $A$ as the finite union $B\cup\bigcup_{m\in A_s}\da(m,s)$ and irreducibility implies that $A$ is already contained in one of the $k+1$ components. Since it can't be contained in $B$, it must be the case that $A_s=\sset{m}$ and $A=\da(m,s)$.
\end{proof}

\subsection{The order sobrification of $\H$}
\label{subsec: sobrification of hoxi}
Recall that the unit~$\eta$ of the order sobrification monad maps $x$ to~$\da x$ and is Scott-continuous. We will employ this in our proof of Theorem~\ref{thm:H=sobH}:
\begin{proof}
  We have the counit map $\varepsilon_{\Gamma(\H)}\colon\Gamma(\H)\to\Gamma(\osob\H)$ discussed in Section~\ref{sec: positive}, Proposition~\ref{prop:union}; it maps a closed set~$C$ to~$\osob
  C=\set{A\in\osob\H}{A\subseteq C}$. Because $\eta_\H$ is Scott-continuous we have the map $\eta^{-1}\colon\Gamma(\osob\H)\to\Gamma(\H)$ in the opposite direction.
  We show that they are inverses of each other. For the first
  calculation let $C$ be a Scott-closed subset of~$\H$.
  \begin{displaymath}
    \begin{array}{rcll}
      x\in\eta^{-1}(\osob C)&\iff&\da x\in\osob C&\mbox{(definition of $\eta$)}\\
      &\iff&\da x\subseteq C&\mbox{(definition of $\osob C$)}\\
      &\iff&x\in C&\mbox{($C$ is a lower set)}
    \end{array}
  \end{displaymath}
  For the other composition, let $\mathcal{C}$ be a Scott-closed subset of~$\osob\H$.
  \begin{displaymath}
    \begin{array}{rcll}
      A\in\mathcal{C}&\stackrel{(*)}{\implies}&\forall x\in A.\;\da x\in\mathcal{C}&\mbox{($\forall x\in A.\;\da x\subseteq A$ as $A$ is a lower set in $\H$,}\\
      &&&\mbox{and $\mathcal{C}$ is a lower set in $\osob\H$)}\\
      &\iff&\forall x\in A.\;\eta(x)\in\mathcal{C}&\mbox{(definition of~$\eta$)}\\
      &\iff&A\subseteq\eta^{-1}(\mathcal{C})\\
      &\iff&A\in\osob{\eta^{-1}(\mathcal{C})}&\mbox{(definition of $\osob{(-)}\;$)}
    \end{array}
  \end{displaymath}
  Our proof will be complete if we can show the reverse of the first
  implication in the calculation above. For this we use our knowledge
  about the elements of~$\osob\H$, that is, the irreducible closed
  subsets of~$\H$ established in the previous section
  (Theorem~\ref{thm:characterisation}). For irreducible subsets~$A$ of
  the form~$\da x$ the reverse of~$(*)$ is trivially true, so assume
  that $A$ is the downset of some level~$L_s$ for $s\in\Nat^*$, and $\forall x\in\da L_s.\;\da x\in\mathcal{C}$. In
  particular we have $\da(m,s)\in\mathcal{C}$ for the elements $(m,s)$ of~$L_s$. Because all elements of level~$L_{\cn ms}$ are below $(m,s)$ (by relation~$<_3$), we have $\da
  L_{\cn ms}\subseteq \da(m,s)\subseteq A$ for all $m\in\Nat$, and hence $\da
  L_{\cn ms}\in\mathcal{C}$ as the latter is downward closed. Finally,
  $A=\da L_s=\cl(\bigcup_{m\in\Nat}\da L_{\cn ms})$ so
  $A\in\mathcal{C}$ follows as desired.
\end{proof}

\section*{Final remarks}
\label{sec: conclusion}
We have shown that the class of dominated dcpos is $\Gamma$-faithful and also seen that it is quite encompassing (Theorem~\ref{thm:dominated}). One may wonder whether there are other natural classes of dcpos on which $\Gamma$ is faithful or whether $\mathbf{domDCPO}$ is in some sense ``maximal.'' Strictly speaking, the answer to this question is no, since the singleton class $\mathcal{C}=\sset{\H}$ is (trivially) $\Gamma$-faithful, yet ---~as we showed~--- not contained in~$\mathbf{domDCPO}$. What is needed, then, is a proper definition of ``maximal'' before any attempt can be made to establish its veracity.

\subsection*{Acknowledgements} This research was started when the first author visited the third in the autumn of 2015, supported by grant SC7-1415-08 from the \emph{London Mathematical Society}. The third author would like to thank the participants of Dagstuhl Seminar 15441 on \emph{Duality in Computer Science} for their comments on an early presentation of the counterexample~$\H$. He is also grateful to Paul Taylor and Xiaodong Jia for discussing and clarifying the categorical setting that gives rise to the order sobrification monad.


\begin{thebibliography}{99}
\bibitem{alghousseinibrahim15}
A. Alghoussein and A. S. Ibrahim. A Counterexample to the Generalized Ho-Zhao
Problem. \emph{Journal of Mathematics Research}, 7(3):137--140, 2015.

\bibitem{abramskyjung94}
S. Abramsky and A. Jung. Domain Theory, volume 3 of Handbook of Logic in Computer Science. Clarendon Press, Oxford, 1994.

\bibitem{drakethron65}
D. Drake and W.J. Thron. On the representation of an abstract lattice as the family
of closed subsets of a topological space. \emph{Trans. Amer. Math. Soc.}, 120:57--71, 1965.

\bibitem{gierzetal03}
G. Gierz, K.H. Hofmann, K. Keimel, J.D. Lawson, M.W. Mislove, and D.S. Scott.
\emph{Continuous Lattices and Domains}. Number 93 in Encyclopedia of Mathematics and
its Applications. Cambridge University Press, Cambridge, 2003.

\bibitem{goubault13a}
J. Goubault-Larrecq. \emph{Non-Hausdorff Spaces and Domain Theory}. Cambridge University Press, Cambridge, 2013.

\bibitem{hoffmann81}
R.-E. Hoffmann. Continuous posets, prime spectra of completely distributive complete
lattices, and {H}ausdorff compactifications. \emph{Lecture Notes in Mathematics}, 871:159--208,
1981.

\bibitem{hozhao09}
W. K. Ho and D. Zhao. Lattices of Scott-closed sets. \emph{Comment. Math. Univ. Carolinae},
50(2):297--314, 2009.

\bibitem{johnstone81}
P. T. Johnstone. Scott is not always sober. \emph{Lecture Notes in Mathematics}, 81:333--334,
1981.

\bibitem{lawson79}
J. D. Lawson. The duality of continuous posets. \emph{Houston Journal of Mathematics},
5:357--394, 1979.

\bibitem{papert59}
S. Papert. Which distributive lattices are lattices of closed sets? \emph{Proceedings of the
Cambridge Philosophical Society}, 55:172--176, 1959. [MR 21:3354].

\bibitem{scott72}
D. S. Scott. Continuous lattices. In E. Lawvere, editor, \emph{Toposes, Algebraic Geometry
and Logic}, volume 274 of \emph{Lecture Notes in Mathematics}, pages 97--136. Springer Verlag, 1972.

\bibitem{thron62}
W. J. Thron. Lattice-equivalence of topological spaces. \emph{Duke Mathematical Journal},
29:671--679, 1962.

\bibitem{zhaofan10}
D. Zhao and T. Fan. Dcpo-completion of posets. \emph{Theor. Comput. Sci.}, 411(22-24):2167--2173, 2010.
\end{thebibliography}

\end{document}